\documentclass[11pt]{amsart}
\pdfoutput=1

\usepackage{amssymb,mathtools}
\usepackage{microtype}

\usepackage[numbers,sort&compress]{natbib}
\usepackage{enumitem}

\usepackage{tikz}

\usepackage{hyperref}
\hypersetup{
  pdftitle={Hodge decomposition and the Shapley value of a cooperative game},
  pdfauthor={Ari Stern and Alexander Tettenhorst},
  pdfsubject={MSC 2010: 91A12, 05C90},  
  bookmarksopen=false,
}
\usepackage{aliascnt}

\theoremstyle{plain} %--default
\newtheorem{theorem}             {Theorem}  [section]

\newaliascnt{corollary}{theorem}
\newtheorem{corollary}[corollary]{Corollary}
\aliascntresetthe{corollary}
\newaliascnt{proposition}{theorem}
\newtheorem{proposition}[proposition]{Proposition}
\aliascntresetthe{proposition}

\theoremstyle{definition}
\newaliascnt{definition}{theorem}
\newtheorem{definition}[definition]{Definition}
\aliascntresetthe{definition}
\newaliascnt{example}{theorem}
\newtheorem{example}[example]{Example}
\aliascntresetthe{example}

\theoremstyle{remark}
\newaliascnt{remark}{theorem}
\newtheorem{remark}[remark]{Remark}
\aliascntresetthe{remark}

\begin{document}
\title[Hodge decomposition and the Shapley value]{Hodge decomposition and the Shapley value of a cooperative game}
\author{Ari Stern}
\address{Department of Mathematics and Statistics, Washington University in St.~Louis}
\email{stern@wustl.edu}

\author{Alexander Tettenhorst}
\email{sasha.tettenhorst@wustl.edu}


\maketitle

\begin{abstract}
  We show that a cooperative game may be decomposed into a sum of
  component games, one for each player, using the combinatorial Hodge
  decomposition on a graph. This decomposition is shown to satisfy
  certain efficiency, null-player, symmetry, and linearity
  properties. Consequently, we obtain a new characterization of the
  classical Shapley value as the value of the grand coalition in each
  player's component game. We also relate this decomposition to a
  least-squares problem involving inessential games (in a similar
  spirit to previous work on least-squares and minimum-norm solution
  concepts) and to the graph Laplacian. Finally, we generalize this
  approach to games with weights and/or constraints on coalition
  formation.
\end{abstract}

\section{Introduction}

In cooperative game theory, one of the central questions is that of
fair division: if players form a coalition to achieve a common goal,
how should they split the profits (or costs) of that achievement among
themselves? (We restrict our attention to transferable utility games,
also called TU games, whose total value may be freely divided and
distributed among the players.) \citet{Shapley1953} introduced one of
the classical solution concepts to this problem, now known as the
\emph{Shapley value}, which he proved to be the unique allocation that
satisfies certain axioms.

In this paper, we show that a cooperative game may be decomposed into
a sum of component games, one for each player, where these components
are uniquely defined in terms of the combinatorial Hodge decomposition
on a hypercube graph associated with the game.  (That is, the value is
apportioned among the players for each possible coalition, not just
the grand coalition consisting of all players.)  We prove that the
Shapley value is precisely the value of the grand coalition in each
player's component game.

This characterization of the game components and the Shapley value
also implies two equivalent characterizations: one in terms of the
least-squares solution to a linear problem, whose solution is exact if
and only if the game is inessential; the other in terms of the graph
Laplacian. The first of these two characterizations is related to the
least-square and minimum-norm solution concepts of \citet{RuVaZa1998}
and \citet{KuSa2007}.

Furthermore, since the combinatorial Hodge decomposition holds for
arbitrary weighted graphs, this decomposition of cooperative games
also generalizes to cases where edges of the hypercube graph are
weighted or removed altogether. This may be seen as modeling variable
willingness or unwillingness of players to join certain coalitions, as
in some models of restricted cooperation. In the latter case, we
compare the resulting solution concepts with other ``Shapley values''
for games with cooperation restrictions, such as the precedence
constraints of \citet{FaKe1992} and the even more general digraph
games of \citet{KhSeTa2016}.

We note that the combinatorial Hodge decomposition has recently been
used to decompose noncooperative games (\citet{CaMeOzPa2011}) and has
also been applied to other problems in economics, such as ranking of
social preferences (\citet{JiLiYaYe2011,HiKaWa2011}). Here, we show
that it can also lend insight to cooperative game theory.

\section{Preliminaries}

\subsection{Cooperative games and the Shapley value}
\label{sec:introGames}

A \emph{cooperative game} consists of a finite set $N$ of players and
a function $ v \colon 2 ^N \rightarrow \mathbb{R} $, which assigns a
value $ v (S) $ to each  coalition $ S \subset N $, such that
$ v ( \emptyset ) = 0 $. Assuming that all players cooperate (forming
the ``grand coalition'' $N$), the question of interest is how to split
the value $ v (N) $ among the players.

The Shapley value $ \phi _i (v) $ allocated to player $ i \in N $ is
based entirely on the marginal value
$ v \bigl( S \cup \{ i \} \bigr) - v (S) $ the player contributes when
joining each coalition $ S \subset N \setminus \{ i \} $. It is
uniquely defined according to the following theorem.

\begin{theorem}[\citet{Shapley1953}]
  \label{thm:shapley}
  There exists a unique allocation
  $ v \mapsto \bigl( \phi _i (v) \bigr) _{ i \in N } $ satisfying the
  following conditions:
  \begin{enumerate}[label=(\alph*)]
  \item efficiency: $ \sum _{ i \in N } \phi _i (v) = v (N) $.
  \item null-player property: If
    $ v \bigl( S \cup \{ i \} \bigr) - v (S) = 0 $ for all
    $ S \subset N \setminus \{ i \} $, then $ \phi _i (v) = 0 $.
  \item symmetry: If
    $ v \bigl( S \cup \{ i \} \bigr) = v \bigl( S \cup \{ j \} \bigr)
    $ for all $ S \subset N \setminus \{ i, j \} $, then
    $ \phi _i (v) = \phi _j (v) $.
  \item linearity: If $ v, v ^\prime $ are two games with the same set
    of players $N$, then
    $ \phi _i ( \alpha v + \alpha ^\prime v ^\prime ) = \alpha \phi _i
    (v) + \alpha ^\prime \phi _i ( v ^\prime ) $ for all
    $ \alpha, \alpha ^\prime \in \mathbb{R} $.
  \end{enumerate}
  Moreover, this allocation is given by the following explicit formula:
  \begin{equation}
    \label{eqn:shapley}
    \phi _i (v) = \sum _{ S \subset N \setminus \{ i \} } \frac{ \lvert S \rvert ! \bigl( \lvert N \rvert - 1 - \lvert S \rvert \bigr) ! }{ \lvert N \rvert ! } \Bigl( v \bigl(  S \cup \{ i \} \bigr) - v (S) \Bigr) .
  \end{equation}
\end{theorem}

The conditions (a)--(d) listed above are often called the
\emph{Shapley axioms}. Simply stated, they say that (a) the value
obtained by the grand coalition is fully distributed among the
players, (b) a player who contributes no marginal value to any
coalition receives nothing, (c) equivalent players receive equal
amounts, and (d) the allocation is linear in the game values.

The formula \eqref{eqn:shapley} has the following useful
interpretation.  Suppose the players form the grand coalition by
joining, one-at-a-time, in the order defined by a permutation $\sigma$
of $N$. That is, player $i$ joins immediately after the coalition
$ S _{ \sigma, i } = \bigl\{ j \in N : \sigma (j) < \sigma (i) \bigr\}
$ has formed, contributing marginal value
$ v \bigl( S _{ \sigma, i } \cup \{ i \} \bigr) - v (S _{ \sigma, i }
) $. Then $ \phi _i (v) $ is the average marginal value contributed by
player $i$ over all $ \lvert N \rvert ! $ permutations $\sigma$, i.e.,
\begin{equation}
  \label{eqn:shapleyPermutation}
  \phi _i (v) = \frac{ 1 }{ \lvert N \rvert ! } \sum _\sigma \Bigl( v \bigl( S _{ \sigma, i } \cup \{ i \} \bigr) - v (S _{ \sigma, i } ) \Bigr) .
\end{equation} 
The equivalence of \eqref{eqn:shapley} and
\eqref{eqn:shapleyPermutation} is due to the fact that
$ \lvert S \rvert ! \bigl( \lvert N \rvert - 1 - \lvert S \rvert
\bigr) ! $ is precisely the number of permutations $\sigma$ for which
$ S = S _{ \sigma, i } $, since there are $ \lvert S \rvert ! $ ways
to permute the preceding players and
$ \bigl( \lvert N \rvert - 1 - \lvert S \rvert \bigr) ! $ ways to
permute the succeeding players.

For purposes of computation, of course, \eqref{eqn:shapley} is
preferable to \eqref{eqn:shapleyPermutation}, since it contains
$ 2 ^{ \lvert N \rvert - 1 } $ terms rather than $ \lvert N \rvert ! $
terms. Computing the Shapley value is \#P-complete (\citet{DePa1994}),
although some recent work has explored polynomial algorithms for
obtaining approximations to the Shapley value
(\citet{CaGoTe2009,CaGoMoTe2017}).

\begin{example}
  \label{ex:introGlove}
  The ``glove game'' is a classic illustrative example of a
  cooperative game. Let $ N = \{ 1, 2, 3 \} $, and suppose that player
  $1$ has a left-hand glove, while players $2$ and $3$ each have a
  right-hand glove. The players wish to put together a pair of gloves,
  which can be sold for value $1$, while unpaired gloves have no
  value. That is, $ v (S) = 1 $ if $S \subset N $ contains both a left
  and a right glove (i.e., player $1$ and at least one of players $2$
  or $3$) and $ v (S) = 0 $ otherwise. The Shapley values for this
  game are
  \begin{equation*}
    \phi _1 (v) = \frac{ 2 }{ 3 } , \qquad \phi _2 (v) = \phi _3 (v) = \frac{ 1 }{ 6 } .
  \end{equation*}
  This is perhaps easiest to interpret from the
  ``average-over-permutations'' perspective: player $1$ contributes
  marginal value $0$ when joining the coalition first (2 of 6
  permutations) and marginal value $1$ otherwise (4 of 6
  permutations), so $ \phi _1 (v) = \frac{ 2 }{ 3 } $. Efficiency and
  symmetry immediately give
  $ \phi _2 (v) = \phi _3 (v) = \frac{ 1 }{ 6 } $.
\end{example}

\subsection{Combinatorial Hodge decomposition and the graph Laplacian}
\label{sec:introHodge}

This section briefly reviews a kind of ``discrete calculus'' on
graphs, which can be used to orthogonally decompose spaces of
functions on vertices and edges. This is a special case of the
\emph{combinatorial Hodge decomposition} on simplicial complexes
(\citet{Eckmann1945}, \citet{Dodziuk1974}), which is itself a discrete
version of the Hodge decomposition on manifolds
(\citet{Hodge1934,Hodge1935,Hodge1936,Hodge1941,Kodaira1949}). While
Hodge theory is a deep subject, connecting algebraic topology with
geometry and elliptic PDE theory, the case of graphs requires nothing
more than finite-dimensional linear algebra. Therefore, we take an
elementary approach here, and we refer readers interested in the more
general theory to the preceding references.

Let $ G = ( V, E ) $ be an oriented graph, where $V$ is the set of
vertices and $E \subset V \times V $ is the set of edges. By
``oriented,'' we mean that at most one of $ ( a, b ) $ and
$ ( b, a ) $ is in $E$ for $ a, b \in V $. If
$ f \colon E \rightarrow \mathbb{R} $ and $ (a,b) \in E $, we define
$ f(b,a) \coloneqq - f (a,b) $ for the reverse-oriented edge.  Denote
by $ \ell ^2 (V) $ the space of functions
$ V \rightarrow \mathbb{R} $, equipped with the $ \ell ^2 $ inner
product
\begin{equation*}
  \langle u , v \rangle \coloneqq \sum _{ a \in V } u (a) v (a) .
\end{equation*}
Similarly, denote by $ \ell ^2 (E) $ the space of functions
$E \rightarrow \mathbb{R} $ with inner product
\begin{equation*}
  \langle f, g \rangle \coloneqq \sum _{ (a,b) \in E } f (a,b) g (a,b) .
\end{equation*}
With respect to the standard bases defined by $V$ and $E$, we can
identify these spaces with $ \mathbb{R} ^{ \lvert V \rvert } $ and
$ \mathbb{R} ^{ \lvert E \rvert } $, each equipped with the Euclidean
dot product.

We next define a linear operator
$ \mathrm{d} \colon \ell ^2 (V) \rightarrow \ell ^2 (E) $ and its
adjoint
$ \mathrm{d} ^\ast \colon \ell ^2 (E) \rightarrow \ell ^2 (V) $, which
are discrete analogs of the gradient and (negative) divergence of
vector calculus (or, more generally, the differential and
codifferential of exterior calculus). Define $ \mathrm{d} $ by
\begin{equation*}
  \mathrm{d} u ( a, b ) \coloneqq u (b) - u (a) .
\end{equation*}
With respect to the standard bases, the matrix of $ \mathrm{d} $ is
just the transpose of the oriented incidence matrix of $G$. The
adjoint is defined by
$ \langle u , \mathrm{d} ^\ast f \rangle \coloneqq \langle \mathrm{d}
u , f \rangle $; explicitly, this is given by
\begin{equation*}
  (\mathrm{d} ^\ast f) (a) = \sum _{ b \sim a } f(b,a),
\end{equation*} 
where $ b \sim a $ denotes that $ (a,b) \in E $ or $ ( b, a ) \in E
$. Again, with respect to the standard bases, the matrix of
$ \mathrm{d} ^\ast $ is just the transpose of that for $ \mathrm{d} $,
i.e., the oriented incidence matrix.

By the fundamental theorem of linear algebra (\citet{Strang1993}), we
may orthogonally decompose $ \ell ^2 (V) $ and $ \ell ^2 (E) $ as
\begin{equation}
  \label{eqn:graphHodgeDecomp}
  \ell ^2 (V) = \mathcal{R} ( \mathrm{d} ^\ast ) \oplus \mathcal{N} ( \mathrm{d} ) , \qquad \ell ^2 (E) = \mathcal{R} ( \mathrm{d} ) \oplus \mathcal{N} ( \mathrm{d} ^\ast ) ,
\end{equation}
where $ \mathcal{R} ( \cdot ) $ and $ \mathcal{N} ( \cdot ) $ denote
range and kernel (nullspace). We call \eqref{eqn:graphHodgeDecomp} the
\emph{combinatorial Hodge decomposition} of $ \ell ^2 (V) $ and
$ \ell ^2 (E) $.

Finally, the \emph{graph Laplacian} is defined by
$ L \coloneqq \mathrm{d} ^\ast \mathrm{d} \colon \ell ^2 (V)
\rightarrow \ell ^2 (V) $. This is identical to the usual graph
Laplacian encountered in, e.g., spectral graph theory
(\citet{Chung1997}), usually expressed as $ L = {D} - A $, where
$ {D} $ is the \emph{degree matrix} and $A$ is the (unsigned)
\emph{adjacency matrix} of the graph $G$. These expressions for $L$
are seen to be identical by observing that
\begin{equation*} 
  (\mathrm{d} ^\ast \mathrm{d} u ) (a)
  = \sum _{ b \sim a } \mathrm{d} u ( b,a ) 
  = \sum _{ b \sim a } \bigl( u(a) - u(b) \bigr) 
  = \operatorname{deg}(a) u (a) - \sum _{ b \sim a } u (b).
\end{equation*} 
The graph Laplacian is intimately related to the Hodge decomposition
in several ways. Of particular interest to us, suppose we wish to
compute the decomposition of $ f \in \ell ^2 (E) $ as
$ \mathrm{d} u + p = f $, where $ u \in \ell ^2 (V) $ and
$ p \in \mathcal{N} ( \mathrm{d} ^\ast ) $. Taking
$ \mathrm{d} ^\ast $ of both sides of this equation gives
$ L u = \mathrm{d} ^\ast f $, so a solution yields the Hodge
decomposition of $f$. This may also be seen as the system of
\emph{normal equations} for the least-squares approximation of $f$ by
$ \mathrm{d} u \in \mathcal{R} (\mathrm{d})$.

\section{Decomposition of cooperative games}
\label{sec:decomp}

\subsection{Cooperative games and the hypercube graph}

Given the set of players $N$, define the oriented graph
$ G = ( V, E ) $ by
\begin{equation*}
  V = 2 ^N , \qquad E = \bigl\{ \bigl(  S, S \cup \{ i \} \bigr) \in V \times V : S \subset N \setminus \{ i \} ,\ i \in N \bigr\} .
\end{equation*}
This is precisely the $ \lvert N \rvert $-dimensional \emph{hypercube
  graph}, where each vertex corresponds to a coalition $ S \subset N$,
and where each edge corresponds to the addition of a single player
$i \notin S $ to $S$, oriented in the direction of the inclusion
$ S \hookrightarrow S \cup \{ i \} $.

With respect to this graph, a cooperative game is precisely a function
$ v \in \ell ^2 (V) $ such that $ v ( \emptyset ) = 0 $. Furthermore,
$ \mathrm{d} v \in \ell ^2 (E) $ gives the marginal value on each
edge, i.e., $ \mathrm{d} v \bigl( S, S \cup \{ i \} \bigr) $ is the
marginal value contributed by player $i$ joining coalition
$S \subset N \setminus \{ i \} $. In order to talk about the marginal
contributions of an individual player $ i \in N $, ignoring those of
the other players $ j \neq i $, we introduce the following collection
of ``partial differentiation'' operators.

\begin{definition}
  \label{def:di}
  For each $ i \in N $, let
  $ \mathrm{d} _i \colon \ell ^2 (V) \rightarrow \ell ^2 (E) $ be the
  operator
  \begin{equation*}
    \mathrm{d} _i u \bigl( S , S \cup \{ j \} \bigr) =
    \begin{cases}
      \mathrm{d} u \bigl( S , S \cup \{ i \} \bigr) & \text{if } i = j, \\
      0 & \text{if } i \neq j .
    \end{cases}
  \end{equation*} 
\end{definition}

Therefore, $ \mathrm{d} _i v \in \ell ^2 (E) $ encodes the marginal
value contributed by player $i$ to the game $v$. For any permutation
$\sigma$ of $N$, which defines a path from $ \emptyset $ to $N$, the
marginal value contributed by player $i$ along this path is
\begin{equation*}
  \sum _{ j \in N } \mathrm{d} _i v \bigl( S _{ \sigma , j } , S _{ \sigma, j } \cup \{ j \} \bigr) = \mathrm{d} v \bigl( S _{ \sigma , i } , S _{ \sigma, i } \cup \{ i \} \bigr) =  v \bigl( S _{ \sigma, i } \cup \{ i \} \bigr) - v (S _{ \sigma, i } ) ,
\end{equation*}
which can be interpreted as a discrete ``line integral'' of
$ \mathrm{d} _i v $ along the path.

\subsection{Decomposition of inessential games}

From \autoref{def:di}, we immediately see that
$ \mathrm{d} = \sum _{ i \in N } \mathrm{d} _i $. However, in general,
$ \mathcal{R} ( \mathrm{d} _i ) \not\subset \mathcal{R} ( \mathrm{d} )
$. To see this, observe that for any permutation $\sigma$,
\begin{equation*}
  \sum _{ j \in N } \mathrm{d} u \bigl( S _{ \sigma , j } , S _{ \sigma, j } \cup \{ j \} \bigr) = \sum _{ j \in N } \Bigl( u \bigl( S _{ \sigma, j } \cup \{ j \} \bigr) - u \bigl( S _{ \sigma , j } \bigr) \Bigr)= u (N) - u ( \emptyset ) ,
\end{equation*}
since the sum telescopes. This value is the same for every permutation
$\sigma$, which is a discrete analog of the fact that the line
integral of a conservative vector field depends only in the endpoints,
not the particular path chosen. Contrast this with the previous
expression: we have already seen that
$ v \bigl( S _{ \sigma, i } \cup \{ i \} \bigr) - v (S _{ \sigma, i }
) $ may be different, depending on the permutation $\sigma$, as in the
glove game of \autoref{ex:introGlove}. The question of when
$ \mathrm{d} _i v \in \mathcal{R} ( \mathrm{d} ) $ is related to the
notion of \emph{inessential games}, as we now show.

\begin{definition}
  The game $v$ is \emph{inessential} if
  $ v (S) = \sum _{ i \in S } v \bigl( \{ i \} \bigr) $ for all
  $ S \subset N $. That is, each coalition obtains the same value
  working together as its individual members would obtain working
  separately.
\end{definition}

\begin{proposition}
  \label{prop:inessential}
  The game $v$ is inessential if and only if
  $ \mathrm{d} _i v \in \mathcal{R} ( \mathrm{d} ) $ for all
  $ i \in N $.
\end{proposition}

\begin{proof}
  If $ \mathrm{d} _i v \in \mathcal{R} ( \mathrm{d} ) $, then from the
  calculation above, it follows that the marginal value
  $ v \bigl( S \cup \{ i \} \bigr) - v ( S ) $ is the same for all
  coalitions $ S \subset N \setminus \{ i \} $. Taking
  $ S = \emptyset $, we see that this value is precisely
  $ v \bigl( \{ i \} \bigr) $. If this holds for all players
  $ i \in N $, then we conclude that $v$ is inessential.

  Conversely, suppose that $v$ is inessential, and define the game
  \begin{equation*}
    v _i (S) =
    \begin{cases}
      v \bigl( \{ i \} \bigr) & \text{if $ i \in S $,}\\
      0 & \text{if $ i \notin S $.}
    \end{cases}
  \end{equation*}
  It follows immediately that
  $ \mathrm{d} _i v = \mathrm{d} v _i \in \mathcal{R} ( \mathrm{d} )
  $, which completes the proof.
\end{proof}

Therefore, inessential games have the decomposition
$ v = \sum _{ i \in N } v _i $, where $ v _i $ is the game described
in the proof above. In the next section, we show how this
decomposition generalizes to arbitrary games, and how the Shapley
value naturally arises from the generalized decomposition.

\subsection{Decomposition of arbitrary games and the Shapley value}

For an arbitrary game $v$, we cannot hope to find games $ v _i $ such
that $ \mathrm{d} _i v = \mathrm{d} v _i $ (unless $v$ is inessential,
as shown in \autoref{prop:inessential}). However, the Hodge
decomposition \eqref{eqn:graphHodgeDecomp} ensures that we \emph{can}
write $ \mathrm{d} _i v \in \ell ^2 (E) $ as the sum of some
$ \mathrm{d} v _i \in \mathcal{R} ( \mathrm{d} ) $ and an element of
$ \mathcal{N} ( \mathrm{d} ^\ast ) $. Moreover, since $G$ is
connected, $ \mathcal{N} ( \mathrm{d} ) \cong \mathbb{R} $, so
$ v _i $ is uniquely determined by the condition
$ v _i ( \emptyset ) = 0 $, i.e., that $v _i $ is itself a game.

The main result of this section, \autoref{thm:gameDecomp}, proves that
these component games $ v _i $ satisfy certain efficiency,
null-player, symmetry, and linearity properties. This means, in
particular, that $ v _i (N) $ satisfies the Shapley axioms, so it must
coincide with the Shapley value $ \phi _i (v) $. In addition to this
axiomatic argument, we will also show directly, in
\autoref{rmk:shapleyFormula}, that $ v _i (N) $ agrees with the
Shapley formula \eqref{eqn:shapley}.

Let $ P \colon \ell ^2 (E) \rightarrow \mathcal{R} ( \mathrm{d} ) $
denote orthogonal projection onto $ \mathcal{R} ( \mathrm{d} ) $, so
that $ P \mathrm{d} _i v $ is the $ \mathcal{R} ( \mathrm{d} ) $
component in the Hodge decomposition of $ \mathrm{d} _i v $.  Note
that this projection is introduced as a theoretical tool; in practice,
we will never need to explicitly construct a matrix for $P$. Instead,
as in \autoref{sec:introHodge}, we will compute the Hodge
decomposition by solving a problem involving the graph Laplacian,
which will be discussed in \autoref{sec:laplacian}.

\begin{theorem}
  \label{thm:gameDecomp}
  For each $ i \in N $, let $ v _i \in \ell ^2 (V) $ with
  $ v _i ( \emptyset ) $ be the unique game such that
  $ \mathrm{d} v _i = P \mathrm{d} _i v $.  Then the games $ v _i $
  satisfy the following:
  \begin{enumerate}[label=(\alph*)]
  \item $ \displaystyle\sum _{ i \in N } v _i = v $.
  \item If $ v \bigl( S \cup \{ i \} \bigr) - v (S) = 0 $ for all
    $ S \subset N \setminus \{ i \} $, then $ v _i = 0 $.
  \item If $ \sigma $ is a permutation of $N$ and $ \sigma ^\ast v $
    is the game
    $ ( \sigma ^\ast v ) (S) = v \bigl( \sigma (S) \bigr) $, then
    $ ( \sigma ^\ast v ) _i = \sigma ^\ast ( v _{ \sigma (i) } ) $. In
    particular, if $ \sigma $ is the permutation swapping $i$ and $j$,
    and if $ \sigma ^\ast v = v $, then
    $ v _i = \sigma ^\ast (v _j) $.

  \item For any two games $ v , v ^\prime $ and
    $ \alpha, \alpha ^\prime \in \mathbb{R} $,
    $ ( \alpha v + \alpha ^\prime v ^\prime ) _i = \alpha v _i +
    \alpha ^\prime v ^\prime _i $.
  \end{enumerate}
  Consequently, $ v _i (N) = \phi _i (v) $ is the Shapley value for
  each player $i \in N $.
\end{theorem}

\begin{proof}
  First, since $ \mathrm{d} = \sum _{ i \in N } \mathrm{d} _i $, we
  have
  \begin{equation*}
    \mathrm{d} \sum _{ i \in N } v _i = \sum _{ i \in N } \mathrm{d} v _i = \sum _{ i \in N } P \mathrm{d} _i v = P \sum _{ i \in N } \mathrm{d} _i v = P \mathrm{d} v = \mathrm{d} v .
  \end{equation*}
  Since $G$ is connected, it follows that $ \sum _{ i \in N } v _i $
  and $v$ differ by a constant. But this constant is
  $ v ( \emptyset ) - \sum _{ i \in N } v _i (\emptyset) = 0 $, which
  proves (a).

  Next, if $ v \bigl( S \cup \{ i \} \bigr) - v (S) = 0 $ for all
  $ S \subset N \setminus \{ i \} $, then $ \mathrm{d} _i v = 0 $. It
  follows that $ \mathrm{d} v _i = P \mathrm{d} _i v = 0 $. Hence,
  $ v _i $ is constant, but since $ v _i ( \emptyset ) = 0 $, we
  conclude that $ v _i = 0 $, which proves (b).

  Next, if $\sigma$ is a permutation of $N$, then direct calculation
  shows that $ \mathrm{d} \sigma ^\ast = \sigma ^\ast \mathrm{d} $ and
  $ \mathrm{d} _i \sigma ^\ast = \sigma ^\ast \mathrm{d} _{ \sigma (i)
  } $. Furthermore, $\sigma$ preserves counting measure and hence the
  $ \ell ^2 $ inner product, so $ P \sigma ^\ast = \sigma ^\ast P
  $. Thus,
  \begin{equation*}
    \mathrm{d} ( \sigma ^\ast v ) _i = P \mathrm{\mathrm{d}} _i (\sigma ^\ast v ) = P \sigma ^\ast ( \mathrm{d} _{ \sigma (i) } v ) = \sigma ^\ast ( P \mathrm{d} _{ \sigma (i) } v ) = \sigma ^\ast ( \mathrm{d} v _{ \sigma (i) } ) = \mathrm{d} \sigma ^\ast ( v _{ \sigma (i) } ) ,
  \end{equation*}
  so $ ( \sigma ^\ast v ) _i $ and
  $ \sigma ^\ast ( v _{ \sigma (i) } ) $ differ by a constant. But
  this constant is
  $ ( \sigma ^\ast v ) _i ( \emptyset ) - \sigma ^\ast ( v _{ \sigma
    (i) } ) ( \emptyset ) = v _i ( \emptyset ) - v _{ \sigma (i) } (
  \emptyset ) = 0 $, which proves (c).

  Next, since $ \mathrm{d} $, $ \mathrm{d} _i $, and $P$ are all linear maps,
  \begin{multline*}
    \mathrm{d} ( \alpha v + \alpha ^\prime v ^\prime ) _i = P \mathrm{d} _i ( \alpha v + \alpha ^\prime v ^\prime ) = \alpha P \mathrm{d} _i v + \alpha ^\prime P \mathrm{d} _i v ^\prime \\
    = \alpha \mathrm{d} v _i + \alpha ^\prime \mathrm{d} v _i ^\prime
    = \mathrm{d} ( \alpha v _i + \alpha ^\prime v _i ^\prime ).
  \end{multline*}
  Hence, the games $ ( \alpha v + \alpha ^\prime v ^\prime ) _i $ and
  $ ( \alpha v _i + \alpha ^\prime v _i ^\prime ) $ differ by a
  constant---but just as above, this constant must be zero, which
  proves (d).

  Finally, having shown (a)--(d), it follows that the allocation
  $ v \mapsto \bigl( v _i (N) \bigr) _{ i \in N } $ satisfies the
  Shapley axioms of \autoref{thm:shapley}. Indeed, condition (a)
  implies efficiency, (b) implies the null-player property, (c)
  implies symmetry, and (d) implies linearity. Hence, by the
  uniqueness property of the Shapley value, we must have that
  $ v _i (N) = \phi _i (v) $ for all $ i \in N $, which completes the
  proof.
\end{proof}

\begin{remark}
  An immediate corollary of \autoref{prop:inessential} and
  \autoref{thm:gameDecomp} is the standard result that
  $ \phi _i (v) = v \bigl( \{ i \} \bigr) $ for all $ i \in N $
  whenever $v$ is inessential.
\end{remark}

\begin{remark}
  \label{rmk:leastSquares}
  Since $P$ is orthogonal projection onto
  $ \mathcal{R} ( \mathrm{d} ) $, we may also view the game $ v _i $
  as the \emph{least-squares solution} to
  $ \mathrm{d}v _i = \mathrm{d} _i v $, which only has an exact
  solution when $v$ is inessential. That is, we have
  \begin{equation*}
    v _i = \operatorname{argmin}\displaylimits_{\substack{u \in \ell ^2 (V) \\ u ( \emptyset ) = 0 }} \lVert \mathrm{d} u - \mathrm{d} _i v \rVert _{ \ell ^2 (E) } .
  \end{equation*}
  This is similar in spirit to the work of \citet{KuSa2007} on
  minimum-norm solution concepts, including the least-square values of
  \citet{RuVaZa1998}. Specifically, \citet{KuSa2007} consider the
  projection of $v$ itself onto the subspace of inessential games in
  $ \ell ^2 (V) $. By contrast, the projection in our approach is
  performed on $ \ell ^2 (E) $.
\end{remark}

\begin{remark}
  \label{rmk:shapleyFormula}
  The decomposition of \autoref{thm:gameDecomp} also gives a
  straightforward way to derive the Shapley formula
  \eqref{eqn:shapley}. By linearity, this formula is equivalent to the
  statement that, if
  $ \chi _{ ( S, S \cup \{ i \} ) } \in \ell ^2 (E) $ is the indicator
  function equal to $1$ on $ \bigl( S, S \cup \{ i \} \bigr) $ and $0$
  on all other edges, and if $ u \in \ell ^2 (V) $ is the solution to
  $ \mathrm{d} u = P \chi _{ ( S, S \cup \{ i \}) } $ with
  $ u (\emptyset) = 0 $, then
  $ u (N) = \frac{ \lvert S \rvert ! ( \lvert N \rvert - 1 - \lvert S
    \rvert ) ! }{ \lvert N \rvert ! } $.

  To prove this, we begin by observing that $ u (N) $ depends only on
  $ \lvert S \rvert $, not on the particular choice of $S$ and
  $i$. Indeed, let $ \lvert T \rvert = \lvert S \rvert $ with
  $ j \notin T $, and let $\sigma$ be a permutation of $N$ such that
  $ \sigma (S) = T $ and $ \sigma (i) = j $. Then
  \begin{equation*}
    \mathrm{d} ( \sigma ^\ast u ) = \sigma ^\ast ( \mathrm{d} u ) = \sigma ^\ast ( P \chi _{ ( S, S \cup \{ i \} ) } ) = P \sigma ^\ast \chi _{ ( S, S \cup \{ i \} ) } = P \chi _{ ( T , T \cup \{ j \} ) },
  \end{equation*}
  and $ ( \sigma ^\ast u ) (N) = u \bigl( \sigma (N) \bigr) = u (N)
  $. Hence, we get the same value for $ u (N) $ whether we choose the
  edge $ \bigl( S, S \cup \{ i \} \bigr) $ or
  $ \bigl( T , T \cup \{ j \} \bigr) $.

  Next, consider the game
  \begin{equation*}
    v (T) \coloneqq
    \begin{cases}
      1 & \text{if } \lvert T \rvert > \lvert S \rvert, \\
      0 & \text{if } \lvert T \rvert \leq \lvert S \rvert,
    \end{cases}
  \end{equation*}
  so that
  \begin{equation*}
    \mathrm{d} v = \sum _{ \substack{\lvert T \rvert = \lvert S \rvert \\ j \notin T } } \chi  _{ ( T, T \cup \{ j \} ) } .
  \end{equation*}
  This sum contains
  $ \binom{\lvert N \rvert }{\lvert S \rvert } \bigl( \lvert N \rvert
  - \lvert S \rvert \bigr) = \frac{ \lvert N \rvert ! }{ \lvert S
    \rvert ! ( \lvert N \rvert - 1 - \lvert S \rvert ) ! } $ terms,
  and the symmetry argument above shows that they all contribute
  equally to $ v (N) $, so
  \begin{equation*}
    v (N) = \frac{ \lvert N \rvert ! }{ \lvert S
      \rvert ! ( \lvert N \rvert - 1 - \lvert S \rvert ) ! } u (N) .
  \end{equation*}
  Finally, since $ v (N) = 1 $, we obtain the claimed expression for
  $ u (N) $.
\end{remark}

\begin{example}
  \label{ex:decompGlove}
  We now illustrate the decomposition of \autoref{thm:gameDecomp} by
  applying it to the glove game introduced in
  \autoref{ex:introGlove}. Since $ \lvert N \rvert = 3 $, the graph
  $G$ consists of vertices and edges of the ordinary,
  three-dimensional cube.

  \autoref{tab:glove} contains the values of $v$ and the component
  games $ v _1 $, $ v _2 $, and $ v _3 $. Several of the properties
  proved in \autoref{thm:gameDecomp} are immediately apparent. In
  particular, we have the decomposition $ v = v _1 + v _2 + v _3 $,
  while $ v _1 (N) = \frac{ 2 }{ 3 } $ and
  $ v _2 (N) = v _3 (N) = \frac{ 1 }{ 6 } $ are the Shapley values
  previously obtained in \autoref{ex:introGlove}. Furthermore, the
  symmetry of players $2$ and $3$ is evident in all three component
  games, not just $ v _2 $ and $ v _3 $. Indeed, if $ \sigma $ is the
  permutation swapping players $2$ and $3$, then
  $ \sigma ^\ast v = v $, so we have
  \begin{equation*}
    v _1 = \sigma ^\ast v _1 , \qquad v _2 = \sigma ^\ast v _3 , \qquad v _3 = \sigma ^\ast v _2 ,
  \end{equation*}
  which can be observed in \autoref{tab:glove}. We also point out
  that, although $ v \geq 0 $, the component games
  $ v _1 , v _2 , v _3 $ may take negative values.
  Note also that $ \mathrm{d} v _i \neq \mathrm{d} _i v $,
  corresponding to the fact that the glove game is not inessential.

  \begin{table}
    \centering
    \renewcommand\arraystretch{1.4}
  \begin{tabular}{crrrr}
    $S$ & $v $ & $ v _1 $ & $ v _2 $ & $ v _3 $ \\
    \hline
    $ \emptyset $ & $0$ & $0$ & $0$ & $0 $ \\
    $ \{ 1 \} $ & $0$ & $ \frac{5}{12} $ & $ - \frac{ 5 }{ 24 } $ & $ - \frac{ 5 }{ 24 } $ \\
    $ \{ 2 \} $ & $0$ & $ - \frac{5}{24} $ & $ \frac{ 1 }{ 6 } $ & $ \frac{ 1 }{ 24 } $ \\
    $ \{ 3 \} $ & $0$ & $ - \frac{5}{24} $ & $ \frac{ 1 }{ 24 } $ & $ \frac{ 1 }{ 6 } $ \\
    $ \{ 1, 2 \} $ & $1$ & $ \frac{ 5 }{ 8 } $ & $ \frac{ 3 }{ 8 } $ & $0$ \\
    $ \{ 1, 3 \} $ & $1$ & $ \frac{ 5 }{ 8 } $ & $0$ & $ \frac{ 3 }{ 8 } $ \\
    $ \{ 2, 3 \} $ & $0$ & $ - \frac{ 1 }{ 4 } $ & $ \frac{ 1 }{ 8 } $ & $ \frac{ 1 }{ 8 } $ \\
    $ \{ 1,2,3 \} $ & $1$ & $ \boldsymbol{ \frac{ 2 }{ 3 } } $ & $ \boldsymbol{ \frac{ 1 }{ 6 } } $ & $ \boldsymbol{ \frac{ 1 }{ 6 } } $
  \end{tabular}
  \bigskip 
  \caption{Decomposition of the three-player glove game as
    $ v = v _1 + v _2 + v _3 $, following
    \autoref{thm:gameDecomp}. The Shapley values of
    $ \frac{ 2 }{ 3 } $, $ \frac{ 1 }{ 6 } $, $ \frac{ 1 }{ 6 } $
    appear in bold on the last line, corresponding to the value of the
    grand coalition in each component game.\label{tab:glove}}
\end{table}
\end{example}

\subsection{Decomposition via the hypercube graph Laplacian}
\label{sec:laplacian}

We now briefly show how the component games $ v _i $ may be computed
using the graph Laplacian $ L = \mathrm{d} ^\ast \mathrm{d} $, without
having to explicitly compute the orthogonal projection operator
$P$. Denote $ L _i \coloneqq \mathrm{d} ^\ast \mathrm{d} _i $; this is
in fact a weighted graph Laplacian, where the edge
$ \bigl( S , S \cup \{ j \} \bigr) $ has weight $1$ if $ i = j $ and
$0$ otherwise. (We will say more about weighted graph Laplacians in
\autoref{sec:weightedDecomp}.)

\begin{proposition}
  \label{prop:laplace}
  For each $ i \in N $, the component game $ v _i $ of
  \autoref{thm:gameDecomp} is the unique solution to
  $ L v _i = L _i v $ such that $ v _i ( \emptyset ) = 0 $.
\end{proposition}

\begin{proof}
  Since
  $ ( \mathrm{d} v _i - \mathrm{d} _i v ) \in \mathcal{N} ( \mathrm{d}
  ^\ast ) $, we immediately have
  \begin{equation*}
    0 = \mathrm{d} ^\ast ( \mathrm{d} v _i - \mathrm{d} _i v ) = \mathrm{d} ^\ast \mathrm{d} v _i - \mathrm{d} ^\ast \mathrm{d} _i v = L v _i - L _i v ,
  \end{equation*}
  so $ L v _i = L _i v $ as claimed. To show uniqueness, suppose that
  $ v _i ^\prime $ is another solution. Then
  $ L ( v _i - v _i ^\prime ) = 0 $, and since the hypercube graph is
  connected, we must have $ v _i - v _i ^\prime $ constant. But
  $ v _i ( \emptyset ) = v _i ^\prime ( \emptyset ) $, so it follows
  that $ v _i = v _i ^\prime $.
\end{proof}

\begin{remark}
  Equivalently, recall from \autoref{rmk:leastSquares} that $ v _i $
  may also be seen as the least-squares solution to
  $ \mathrm{d} v _i = \mathrm{d} _i v $. From this point of view,
  $ \mathrm{d} ^\ast \mathrm{d} v _i = \mathrm{d} ^\ast \mathrm{d} _i
  v $ is precisely the system of \emph{normal equations} corresponding
  to this least-squares problem, as previously discussed in
  \autoref{sec:introHodge}.
\end{remark}

\subsection{Explicit decomposition via discrete Green's functions}

In \autoref{prop:laplace}, we characterized the component game
$ v _i $ implicitly, as the solution to $ L v _i = L _i v $ with
$ v _i ( \emptyset ) = 0 $. In fact, it is possible to obtain an
\emph{explicit} formula for $ v _i $, in terms of the \emph{discrete
  Green's function} for the hypercube graph Laplacian (\citet[Example
4]{ChYa2000}). For an arbitrary game $ v $ and coalition
$ S \subset N $, the formula for $ v _i (S) $ is a bit unwieldy, but
we also present two situations in which it simplifies nicely. First,
when $ S = N $, we show that the formula for $ v _i (N) $ yields the
Shapley formula \eqref{eqn:shapley}, giving an alternative proof that
$ v _i (N) = \phi _i (v) $. Second, when $v$ is a \emph{pure
  bargaining game}, where $ v (S) = 0 $ for $ S \neq N $, we give an
explicit formula for $ v _i (S) $.

Let $ y \in \mathcal{R} ( \mathrm{d} ^\ast ) $, and consider the
problem $ L u = y $.  Rather than imposing the condition
$ u ( \emptyset ) = 0 $, we instead seek the unique solution
$u \in \mathcal{R} ( \mathrm{d} ^\ast ) $.\footnote{Since
  $ \mathcal{R} ( \mathrm{d} ^\ast ) $ is the orthogonal complement of
  $ \mathcal{N} ( \mathrm{d} ) \cong \mathbb{R} $, this is equivalent
  to saying that $y$ and $u$ are orthogonal to constant functions,
  i.e., that
  $ \sum _{ S \subset N } y (S) = \sum _{ S \subset N } u (S) = 0 $.}
The solution vanishing on $ \emptyset $ is then simply
$ u - u ( \emptyset ) $.  Since the restriction of $L$ to
$ \mathcal{R} ( \mathrm{d} ^\ast ) $ is symmetric positive definite,
it has a symmetric positive definite inverse $K$, and the solution is
$ u = K y $. Writing out the matrix multiplication explicitly, we have
\begin{equation*}
  u (S) = \sum _{ T \subset N } K ( S, T ) y (T) ,
\end{equation*}
where $ K ( \cdot , \cdot ) $ is called the \emph{discrete Green's
  function}. \citet[Example 4]{ChYa2000} showed that, for the
hypercube graph Laplacian with $ \lvert N \rvert = n $, this is given
by the formula
\begin{multline*}
  K ( S, T ) = \frac{ 1 }{ n 2 ^{ 2 n } } \biggl( - \sum _{ j < k }
  \frac{ \bigl[ \binom{n}{0} + \cdots + \binom{n}{j}\bigr] \bigl[
    \binom{n}{j+1} + \cdots + \binom{n}{n} \bigr] }{ \binom{n-1}{j} }\\
  + \sum _{ k \leq j } \frac{ \bigl[ \binom{n}{j+1} + \cdots +
    \binom{n}{n} \bigr] ^2 }{ \binom{n-1}{j} } \biggr),
\end{multline*}
where $k = \lvert S \mathbin{\triangle} T \rvert $ is the distance
between $S$ and $T$ in the graph; here,
$ S \mathbin{\triangle} T \coloneqq ( S \cup T ) \setminus ( S \cap T
) $ is the symmetric difference of $S$ and $T$. (We remark that this
differs from the formula in \citet{ChYa2000} by a factor of
$\frac{1}{n}$, since they consider the normalized graph Laplacian
$ \frac{1}{n} L $.)

Given a game $v$, we may therefore use this approach to calculate
$ u _i = K L _i v $ explicitly and then take
$ v _i = u _i - u _i ( \emptyset ) $. Fortunately, the fact that
$ y = L _i v $ simplifies the formulas substantially. Given
$ T \subset N \setminus \{ i \} $,
\begin{equation*}
  ( L _i v ) \bigl( T \cup \{ i \} \bigr) = v \bigl( T \cup \{ i \} \bigr) - v (T) = - ( L _i v ) (T) ,
\end{equation*}
and therefore,
\begin{equation*}
  u _i (S) = \sum _{ T \subset N \setminus \{ i \} } \Bigl( K \bigl( S, T \cup \{ i \} \bigr)  - K ( S, T ) \Bigr) \Bigl( v \bigl( T \cup \{ i \} \bigr) - v (T) \Bigr) 
\end{equation*} 
Furthermore, if $ S \subset N \setminus \{ i \} $ is distance $k$ from
$T$, then it is distance $ k + 1 $ from $ T \cup \{ i \} $. Therefore,
all but the $ j = k $ terms in
$ K \bigl( S, T \cup \{ i \} \bigr) - K ( S, T ) $ cancel, leaving
\begin{multline*}
  K \bigl( S , T \cup \{ i \} \bigr) - K ( S, T )\\
  \begin{aligned}
  &= - \frac{ 1 }{ n 2 ^{ 2 n } } \frac{ \bigl[ \binom{n}{0} + \cdots + \binom{n}{k}\bigr] \bigl[
    \binom{n}{k+1} + \cdots + \binom{n}{n} \bigr] + \bigl[ \binom{n}{k+1} + \cdots + \binom{n}{n} \bigr] ^2 }{ \binom{n-1}{k} }\\
  &= - \frac{ 1 }{ n 2 ^{ 2 n } } \frac{ \bigl[ \binom{n}{0} + \cdots + \binom{n}{k} + \binom{n}{k+1} + \cdots + \binom{n}{n} \bigr] \bigl[ \binom{n}{k+1} + \cdots + \binom{n}{n} \bigr] }{ \binom{n-1}{k}  } \\
  &= - \frac{ 1 }{ n 2 ^n} \frac{ \binom{n}{k+1} + \cdots + \binom{n}{n} }{ \binom{n-1}{k}  } .
\end{aligned}
\end{multline*}
Similarly, $ S \cup \{ i \} $ is distance $k$ from $ T \cup \{ i \} $
and distance $ k + 1 $ from $ T $, so
\begin{equation*}
  K \bigl( S \cup \{ i \} , T \cup \{ i \} \bigr) - K \bigl(  S \cup \{ i \} , T \bigr) = \frac{ 1 }{ n 2 ^n} \frac{ \binom{n}{k+1} + \cdots + \binom{n}{n} }{ \binom{n-1}{k}  }.
\end{equation*}
We have just proved the following theorem.

\begin{theorem}
  \label{thm:explicit}
  For each $ i \in N $, the component game $ v _i $ of
  \autoref{thm:gameDecomp} is given by
  $ v _i = u _i - u _i ( \emptyset ) $, where, for
  $ S \subset N \setminus \{ i \} $,
  \begin{equation*} 
    u _i \bigl( S \cup \{ i \} \bigr)  = \frac{ 1 }{ \lvert N \rvert  2 ^{\lvert N \rvert }} \sum _{ T \subset N \setminus \{ i \} } \frac{ \binom{\lvert N \rvert}{\lvert S \mathbin{\triangle} T \rvert +1} + \cdots + \binom{\lvert N \rvert }{\lvert N \rvert } }{ \binom{\lvert N \rvert -1}{\lvert S \mathbin{\triangle} T \rvert }  } \Bigl( v \bigl( T \cup \{ i \} \bigr) - v (T) \Bigr),
  \end{equation*}
  and $ u _i (S) = - u _i \bigl( S \cup \{ i \} \bigr) $. In particular,
  \begin{equation*}
    u _i ( \emptyset ) = - \frac{ 1 }{ \lvert N \rvert  2 ^{\lvert N \rvert } } \sum _{ T \subset N \setminus \{ i \} } \frac{ \binom{\lvert N \rvert }{\lvert T \rvert +1} + \cdots + \binom{\lvert N \rvert }{ \lvert N \rvert } }{ \binom{\lvert N \rvert -1}{ \lvert T \rvert }  } \Bigl( v \bigl( T \cup \{ i \} \bigr) - v (T) \Bigr).
  \end{equation*}
\end{theorem}

\begin{remark}
  This immediately gives us an explicit formula for the projection
  $ P \mathrm{d} _i v = \mathrm{d} v _i = \mathrm{d} u _i \in
  \mathcal{R} ( \mathrm{d} ) $.
\end{remark}

We now use this formula to get a new proof that
$ v _i (N) = \phi _i (v) $ is the Shapley value. Observe that the
distance from $ S = N \setminus \{ i \} $ to
$ T \subset N \setminus \{ i \} $ is just
$ \lvert N \rvert - \lvert T \rvert - 1 $, so \autoref{thm:explicit} says
that
\begin{align*}
  u _i (N)
  &= \frac{ 1 }{ \lvert N \rvert  2 ^{\lvert N \rvert } } \sum _{ T \subset N \setminus \{ i \} } \frac{ \binom{\lvert N \rvert }{\lvert N \rvert - \lvert T \rvert } + \cdots + \binom{\lvert N \rvert }{ \lvert N \rvert } }{ \binom{\lvert N \rvert -1}{\lvert N \rvert - \lvert T \rvert - 1 }  } \Bigl( v \bigl( T \cup \{ i \} \bigr) - v (T) \Bigr) \\
  &= \frac{ 1 }{ \lvert N \rvert  2 ^{\lvert N \rvert } } \sum _{ T \subset N \setminus \{ i \} } \frac{ \binom{\lvert N \rvert }{ 0 }+ \cdots + \binom{\lvert N \rvert }{\lvert T \rvert }  }{ \binom{\lvert N \rvert -1}{\lvert T \rvert }  } \Bigl( v \bigl( T \cup \{ i \} \bigr) - v (T) \Bigr).
\end{align*}
Subtracting $ v _i (N) = u _i (N) - u _i ( \emptyset ) $ and using
$ \binom{\lvert N \rvert }{0} + \cdots + \binom{\lvert N \rvert }{
  \lvert N \rvert } = 2 ^{ \lvert N \rvert } $ gives
\begin{align*}
  v _i (N) &= \frac{ 1 }{ \lvert N \rvert } \sum _{ T \subset N \setminus \{ i \} } \frac{ 1 }{ \binom{\lvert N \rvert -1}{\lvert T \rvert }  } \Bigl( v \bigl( T \cup \{ i \} \bigr) - v (T) \Bigr) \\
           &= \sum _{ T \subset N \setminus \{ i \} } \frac{ \lvert T \rvert ! \bigl( \lvert N \rvert - 1 - \lvert T \rvert \bigr) ! }{ \lvert N \rvert ! } \Bigl( v \bigl( T \cup \{ i \} \bigr) - v (T) \Bigr),
\end{align*}
which is precisely the Shapley formula \eqref{eqn:shapley} for
$ \phi _i (v) $.

Finally, we apply \autoref{thm:explicit} to the case where $v$ is a
pure bargaining game, i.e., where $ v (S) = 0 $ for $ S \neq N $.

\begin{theorem}
  Let $v$ be a pure bargaining game, $ i \in N $, and
  $ S \subset N \setminus \{ i \} $. Then the component game $ v _i $
  is given by
  \begin{align*}
    v _i \bigl( S \cup \{ i \} \bigr) &= \frac{ 1 }{ \lvert N \rvert 2 ^{ \lvert N \rvert } } \Biggl( 1 + \frac{ \binom{\lvert N \rvert }{0} + \cdots + \binom{\lvert N \rvert }{ \lvert S \rvert } }{ \binom{\lvert N \rvert - 1 }{ \lvert S \rvert } } \Biggr) v (N), \\
    v _i (S) &= \frac{ 1 }{ \lvert N \rvert 2 ^{ \lvert N \rvert } } \Biggl( 1 - \frac{ \binom{\lvert N \rvert }{0} + \cdots + \binom{\lvert N \rvert }{ \lvert S \rvert } }{ \binom{\lvert N \rvert - 1 }{ \lvert S \rvert } } \Biggr) v (N)     .
  \end{align*}
\end{theorem}

\begin{proof}
  Observe that $ v \bigl( T \cup \{ i \} \bigr) - v (T) $ vanishes for
  $ T \subsetneq N \setminus \{ i \} $ and equals $ v (N) $ for
  $ T = N \setminus \{ i \} $. The distance between
  $S \subset N \setminus \{ i \} $ and $ T = N \setminus \{ i \} $ is
  just $ \lvert N \rvert - \lvert S \rvert - 1 $, so applying
  \autoref{thm:explicit} gives
  \begin{align*}
    u _i \bigl( S \cup \{ i \} \bigr)
    &= \frac{ 1 }{ \lvert N \rvert  2 ^{\lvert N \rvert } } \frac{ \binom{\lvert N \rvert }{\lvert N \rvert - \lvert S \rvert } + \cdots + \binom{\lvert N \rvert }{ \lvert N \rvert } }{ \binom{\lvert N \rvert -1}{\lvert N \rvert - \lvert S \rvert - 1 }  } v (N) \\
    &= \frac{ 1 }{ \lvert N \rvert  2 ^{\lvert N \rvert } } \frac{ \binom{\lvert N \rvert }{ 0 } + \cdots + \binom{\lvert N \rvert }{ \lvert S \rvert } }{ \binom{\lvert N \rvert -1}{\lvert S \rvert }  } v (N) ,
  \end{align*}
  and $ u _i (S) = - u _i \bigl( S \cup \{ i \} \bigr) $.  In
  particular, for $ S = \emptyset $, we get
  \begin{equation*}
    u _i ( \emptyset ) = - \frac{ 1 }{ \lvert N \rvert  2 ^{\lvert N \rvert } } v (N),
  \end{equation*}
  and the result follows immediately from
  $ v _i = u _i - u _i ( \emptyset ) $.
\end{proof}

\section{Weighted decompositions and restricted cooperation}

\subsection{Decomposition of cooperative games with weighted edges}
\label{sec:weightedDecomp}

Suppose that each edge $ \bigl( S, S \cup \{ i \} \bigr) \in E $ of
the hypercube graph is assigned a weight
$ w \bigl( S, S \cup \{ i \} \bigr) > 0 $. We define
$ \ell ^2 _w (E) $ to be the space of functions
$E \rightarrow \mathbb{R} $ equipped with the weighted $ \ell ^2 $
inner product,
\begin{equation*}
  \langle f, g \rangle _w  \coloneqq \sum _{ ( S, S \cup \{ i \} ) \in E } w \bigl( S, S \cup \{ i \} \bigr) f \bigl( S, S \cup \{ i \} \bigr) g \bigl( S, S \cup \{ i \} \bigr).
\end{equation*} 
The setting of \autoref{sec:decomp} corresponds to the special case
where $ w = 1 $ on every edge. (Equivalently, $w$ may be any positive
constant, not necessarily $1$.)

The weighted inner product affects the Hodge decomposition as
follows. Although
$ \mathrm{d} \colon \ell ^2 (V) \rightarrow \ell ^2 _w (E) $ is
unchanged,
$ \mathrm{d} ^\ast _w \colon \ell ^2 _w (E) \rightarrow \ell ^2 (V) $
is now the adjoint with respect to the weighted inner product. We then
have the combinatorial Hodge decomposition
\begin{equation*}
  \ell ^2 _w (E) = \mathcal{R} (\mathrm{d}) \oplus \mathcal{N} ( \mathrm{d} _w ^\ast ) , 
\end{equation*} 
where the direct sum is $ \ell ^2 _w $-orthogonal rather than
$ \ell ^2 $-orthogonal.

Denote by
$ P _w \colon \ell ^2 _w (E) \rightarrow \mathcal{R} ( \mathrm{d} ) $
the $ \ell ^2 _w $-orthogonal projection onto
$ \mathcal{R} ( \mathrm{d} ) $.  The decomposition of cooperative
games in \autoref{thm:gameDecomp} may then be generalized as follows.

\begin{theorem}
  \label{thm:weightedDecomp}
  For each $ i \in N $, let $ v _{i, w}  \in \ell ^2 (V) $ with
  $ v _{i,w} ( \emptyset ) = 0 $ be the unique game such that
  $ \mathrm{d} v _{i,w} = P _w \mathrm{d} _i v $. Then:
  \begin{enumerate}[label=(\alph*)]
  \item $ \displaystyle\sum _{ i \in N }v _{i,w}  = v $.
  \item If $ v \bigl( S \cup \{ i \} \bigr) - v (S) = 0 $ for all
    $ S \subset N \setminus \{ i \} $, then $ v _{i,w} = 0 $.
  \item If $ \sigma $ is a permutation of $N$, then
    $ ( \sigma ^\ast v ) _{ i , \sigma ^\ast w } = \sigma ^\ast ( v _{
      \sigma (i) , w } ) $. In particular, if $\sigma$ is the
    permutation swapping $i$ and $j$, and if $ \sigma ^\ast v = v $
    and $ \sigma ^\ast w = w $, then $ v _{i,w} = \sigma ^\ast ( v _{j,w} ) $.
  \item For any two games $ v, v ^\prime $ and $ \alpha , \alpha ^\prime \in \mathbb{R}  $, $ ( \alpha v + \alpha ^\prime v ^\prime ) _{ i , w } = \alpha v _{ i , w } + \alpha ^\prime v ^\prime _{ i , w } $.
  \end{enumerate}
\end{theorem}

\begin{proof}
  The proofs of (a), (b), and (d) are just as in
  \autoref{thm:gameDecomp}, since the weighted projection $ P _w $ is
  still linear and equal to the identity on
  $ \mathcal{R} ( \mathrm{d} ) $.

  For (c), we can no longer assume that a permutation preserves the
  $ \ell ^2 _w $ inner product. However, we do have
  $ \langle f, g \rangle _w = \langle \sigma ^\ast f , \sigma ^\ast g
  \rangle _{ \sigma ^\ast w } $, which implies
  $ P _{ \sigma ^\ast w } \sigma ^\ast = \sigma ^\ast P _w $. Therefore,
  \begin{multline*}
    \mathrm{d} ( \sigma ^\ast v ) _{ i , \sigma ^\ast w } = P _{ \sigma ^\ast w } \mathrm{d} _i ( \sigma ^\ast v ) = P _{ \sigma ^\ast w } \sigma ^\ast ( \mathrm{d} _{ \sigma (i) } v ) \\
    = \sigma ^\ast ( P _w \mathrm{d} _{ \sigma (i) } v ) = \sigma ^\ast ( \mathrm{d} v _{ \sigma (i) , w } ) = \mathrm{d} \sigma ^\ast ( v _{ \sigma (i) , w } ),
  \end{multline*}
  and the rest of the argument proceeds as in the proof of
  \autoref{thm:gameDecomp}.
  \end{proof}
  
  \begin{corollary}
    \label{cor:weightedShapley}
    If $ \sigma ^\ast w = w $ for all permutations $\sigma$, then
    $ \phi _i (v) = v _{ i , w } (N) $.
  \end{corollary}
  
  \begin{proof}
    If $w$ is invariant under permutations, then (a)--(d) imply that
    $ v _{ i , w } (N) $ satisfies the Shapley axioms, so it must be the
    Shapley value $ \phi _i (v) $.
  \end{proof}

  As in \autoref{rmk:leastSquares}, we may view the component $ v _i $
  as a \emph{weighted} least-squares solution to
  $ \mathrm{d} v _i = \mathrm{d} _i v $, in the sense that
  \begin{equation*}
    v _i = \operatorname{argmin}\displaylimits_{\substack{u \in \ell ^2 (V) \\ u ( \emptyset ) = 0 }} \lVert \mathrm{d} u - \mathrm{d} _i v \rVert _{ \ell ^2 _w (E) } .
  \end{equation*}
  We can also cast this in terms of the weighted graph Laplacians
  $ L _w \coloneqq \mathrm{d} ^\ast _w \mathrm{d} $ and
  $ L _{ w _i } \coloneqq \mathrm{d} ^\ast _w \mathrm{d} _i =
  \mathrm{d} ^\ast _{ w _i } \mathrm{d} $, where the weight function
  $ w _i $ is defined by
  \begin{equation*}
    w _i \bigl( S, S \cup \{ j \} \bigr) \coloneqq
    \begin{cases}
      w \bigl( S, S \cup \{ i \} \bigr) & \text{if } i = j ,\\
      0 & \text{if } i \neq j .
    \end{cases}
  \end{equation*}
  The following generalization of \autoref{prop:laplace} is stated
  without proof, since the proof is essentially identical. It can also
  be seen as an expression of the normal equations for the weighted
  least-squares problem.

  \begin{proposition}
  \label{prop:weightedLaplace}
  For each $ i \in N $, the component game $ v _i $ of
  \autoref{thm:weightedDecomp} is the unique solution to
  $ L _w v _i = L _{w _i} v $ such that $ v _i ( \emptyset ) = 0 $.
\end{proposition}

  One consequence of \autoref{cor:weightedShapley} is that, although
  the Shapley value is unique, the decomposition
  $ v = \sum _{ i \in N } v _i $ of \autoref{thm:gameDecomp} generally
  is not. Indeed, any totally symmetric weight function $w$ will yield
  a decomposition satisfying the conditions of
  \autoref{thm:gameDecomp}, and these will generally not agree with
  one another---except at $N$, where they all give the Shapley
  value. This is illustrated in the following example.

\begin{example}
  \label{ex:symmetricGlove}

  In \autoref{ex:decompGlove}, we decomposed the glove game with
  respect to the $ \ell ^2 $ inner product, corresponding to the
  constant weight function $ w \equiv 1 $.  Suppose instead that we
  take $ w \bigl( S, S \cup \{ i \} \bigr) = \lvert S \rvert + 1 $,
  which is totally symmetric but not constant. The resulting decomposition is
  shown in \autoref{tab:symmetricGlove}, and is distinct from that
  obtained in \autoref{ex:decompGlove} and shown in
  \autoref{tab:glove}. However, due to the symmetry of $w$, the
  Shapley values are again recovered as the value of the grand
  coalition in each component game. Note that the symmetry of players
  $2$ and $3$ is still apparent in the component games.

  \begin{table}
    \centering
    \renewcommand\arraystretch{1.4}
  \begin{tabular}{crrrr}
    $S$ & $v $ & $ v _1 $ & $ v _2 $ & $ v _3 $ \\
    \hline
    $ \emptyset $ & $0$ & $0$ & $0$ & $0 $ \\
    $ \{ 1 \} $ & $0$ & $ \frac{ 16 }{ 31 } $ & $ -\frac{ 8 }{ 31 } $ & $ -\frac{ 8 }{ 31 } $ \\
    $ \{ 2 \} $ & $0$ & $ -\frac{ 8 }{ 31 } $ & $ \frac{ 6 }{ 31 } $ &  $ \frac{ 2 }{ 31 } $ \\
    $ \{ 3 \} $ & $0$ & $ - \frac{ 8 }{ 31 } $ & $ \frac{ 2 }{ 31 } $ &  $ \frac{ 6 }{ 31 } $ \\
    $ \{ 1, 2 \} $ & $1$ & $ \frac{ 20 }{ 31 } $ & $ \frac{ 21 }{ 62 } $ &  $ \frac{ 1 }{ 62 } $ \\
    $ \{ 1, 3 \} $ & $1$ & $ \frac{ 20 }{ 31 } $ & $ \frac{ 1 }{ 62 } $ &  $ \frac{ 21 }{ 62 } $ \\
    $ \{ 2, 3 \} $ & $0$ & $ -\frac{ 9 }{ 31 } $ & $ \frac{ 9 }{ 62 } $ &  $ \frac{ 9 }{ 62 } $ \\
    $ \{ 1,2,3 \} $ & $1$ & $ \boldsymbol{ \frac{ 2 }{ 3 } } $ & $ \boldsymbol{ \frac{ 1 }{ 6 } } $ & $ \boldsymbol{ \frac{ 1 }{ 6 } } $
  \end{tabular}
  \bigskip 
  \caption{%
    Decomposition of the three-player glove game as
    $ v = v _1 + v _2 + v _3 $, following
    \autoref{thm:weightedDecomp}, with weight function
    $ w (S) = \lvert S \rvert + 1 $. The Shapley values of
    $ \frac{ 2 }{ 3 } $, $ \frac{ 1 }{ 6 } $, $ \frac{ 1 }{ 6 } $
    appear in bold on the last line, since the weight function is
    totally symmetric, but the components are elsewhere distinct from those in
    \autoref{tab:glove}.\label{tab:symmetricGlove} }
\end{table}
\end{example}
  
\begin{example}
  Again, consider the glove game, but take the weight function to be
  $ w \bigl( \emptyset , \{ 1 \} \bigr) = \frac{1}{2} $ and $ w = 1 $
  otherwise. This may be interpreted as player $1$ being reluctant
  (but not totally unwilling) to be the first player to join the
  coalition. The resulting decomposition is shown in
  \autoref{tab:reluctantGlove}. Unlike in the previous examples, this
  $w$ is not totally symmetric, and consequently the values
  $ v _1 (N) = \frac{ 13 }{ 17 } $ and
  $ v _2 (N) = v _3 (N) = \frac{ 2 }{ 17 } $ no longer agree with the
  Shapley values. Since player $1$ is less willing to join the
  coalition first (i.e., to contribute zero marginal value), the
  payoff to player $1$ is increased from $ \frac{ 2 }{ 3 } $ to
  $ \frac{ 13 }{ 17 } $ at the expense of players $2$ and $3$, the
  payoff to each of whom is reduced from $ \frac{ 1 }{ 6 } $ to
  $ \frac{ 2 }{ 17 } $. Note that the symmetry of players $2$ and $3$
  is still maintained.

  \begin{table}
    \centering
    \renewcommand\arraystretch{1.4}
  \begin{tabular}{crrrr}
    $S$ & $v $ & $ v _1 $ & $ v _2 $ & $ v _3 $ \\
    \hline
    $ \emptyset $ & $0$ & $0$ & $0$ & $0 $ \\
    $ \{ 1 \} $ & $0$ & $ \frac{ 10 }{ 17 } $ & $ - \frac{ 5 }{ 17 } $ & $ - \frac{ 5 }{ 17 } $ \\
    $ \{ 2 \} $ & $0$ & $ - \frac{ 5 }{ 34 } $ & $ \frac{ 37 }{ 272 } $ & $ \frac{ 3 }{ 272 } $ \\
    $ \{ 3 \} $ & $0$ & $ - \frac{ 5 }{ 34 } $ & $ \frac{ 3 }{ 272 } $ & $ \frac{ 37 }{ 272 } $ \\
    $ \{ 1, 2 \} $ & $1$ & $ \frac{ 25 }{ 34 } $  & $  \frac{ 87 }{ 272 } $ & $ - \frac{ 15 }{ 272 } $\\
    $ \{ 1, 3 \} $ & $1$ & $ \frac{ 25 }{ 34 } $ & $ - \frac{ 15 }{ 272 } $ & $  \frac{ 87 }{ 272 } $ \\
    $ \{ 2, 3 \} $ & $0$ & $ - \frac{ 3 }{ 17 } $ & $ \frac{ 3 }{ 34 } $ & $ \frac{ 3 }{ 34 } $ \\
    $ \{ 1,2,3 \} $ & $1$ & $ \boldsymbol{ \frac{ 13 }{ 17 }  } $ & $ \boldsymbol{ \frac{ 2 }{ 17 } } $ & $ \boldsymbol{ \frac{ 2 }{ 17 }  } $
  \end{tabular}
  \bigskip 
  \caption{ Decomposition of the three-player glove game as
    $ v = v _1 + v _2 + v _3 $, following
    \autoref{thm:weightedDecomp}, with weight function
    $ w \bigl( \emptyset , \{ 1 \} \bigr) = \frac{1}{2} $ and
    $ w = 1 $ otherwise. The bold values $ v _i (N) $ on the last line
    no longer correspond to the Shapley values, since $w$ is not
    totally symmetric.\label{tab:reluctantGlove} }
\end{table}
\end{example}

\begin{remark}
  The weighted Shapley value of \citet{Shapley1953a} and
  \citet{KaSa1987} also models asymmetry between players, but it does
  so in a fundamentally different way from the approach considered
  here. In \citet{Shapley1953a} each \emph{player} is assigned a
  weight, while \citet{KaSa1987} generalized these player weights to
  ``weight systems.'' It would be interesting to investigate whether
  these approaches can be related by using player weights (or weight
  systems) to construct corresponding edge weights.
\end{remark}

\subsection{Decomposition of games with restricted cooperation}

The framework discussed in the preceding sections, as in
\citet{Shapley1953}, assumes that every player $ i \in N $ is willing
to join every coalition $ S \subset N $, so every such coalition may
be feasibly formed \emph{en route} to the grand coalition. In models
of restricted cooperation, however, this is not the case. The
\emph{precedence constraints} of \citet{FaKe1992} impose a partial
ordering on $N$, so that some players are constrained to join the
coalition prior to others. \citet{KhSeTa2016} have recently
generalized this to so-called digraph games, where precedence is
determined by a digraph on $N$ that (unlike the \citet{FaKe1992}) may
contain cycles; a player $i$ may be required to precede another player
$j$ in some coalitions but not others. (For another recent model of
restricted cooperation, see \citet{KoSuTa2017}.)

The constraints above all correspond to situations where a player $i$
is forbidden to join a coalition $ S \subset N \setminus \{ i \} $. In
this case, we say that the edge $ \bigl( S , S \cup \{ i \} \bigr) $
is \emph{infeasible}, and we remove it from the hypercube graph. If we
continue in this manner, removing all edges and vertices that are
incompatible with the constraints, then we arrive at a graph
$ G = ( V , E ) $ which is a \emph{subgraph} of the hypercube
graph. Here, $V$ contains the so-called \emph{feasible coalitions}
that are compatible with the constraints.

Assume that $G$ is connected and that $ \emptyset , N \in V $, so that
a coalition is feasible if and only if it can be formed starting from
$ \emptyset $, and the grand coalition is feasible. Since the Hodge
decomposition may be defined on any graph---in particular, on the
subgraph $G$ of the hypercube graph---we again obtain a decomposition
$ v = \sum _{ i \in N } v _i $, defined by
$ \mathrm{d} v _i = P \mathrm{d} _i v $ with
$ v _i ( \emptyset ) = 0 $ for $ i \in N $. Since $G$ is connected, we
again have $ \mathcal{N} ( \mathrm{d} ) \cong \mathbb{R} $, so the
decomposition is unique; moreover, it satisfies conditions (a)--(d) of
\autoref{thm:weightedDecomp}, if we interpret the missing edges as
having weight zero.

\begin{example}
  In the glove game, suppose that player $1$ refuses to join the
  coalition first, so that $ \{ 1 \} $ and all its incident edges are
  removed from the graph. The resulting decomposition is shown in
  \autoref{tab:holdout1Glove}. Note that $ v _1 = v $ and
  $ v _2 = v _3 = 0 $, so that player $1$ captures \emph{all} of the
  value of the game. The reason for this is that, by removing the only
  edges on which players $2$ and $3$ contribute marginal value, the
  constraints have turned players $2$ and $3$ into null
  players. Observe also that $ \mathrm{d} v _i = \mathrm{d} _i v $ for
  all $ i \in N $, so the constraints have effectively made the game
  inessential.
  
  \begin{table}
    \centering
    \renewcommand\arraystretch{1.4}
    \begin{tabular}{crrrr}
      $S$ & $v $ & $ v _1 $ & $ v _2 $ & $ v _3 $ \\
      \hline
      $ \emptyset $ & $0$ & $0$ & $0$ & $0 $ \\
      $ \{ 2 \} $ & $0$ & $0$ & $0$ & $0$\\
      $ \{ 3 \} $ & $0$ & $0$ & $0$ & $0$ \\
      $ \{ 1, 2 \} $ & $1$ & $1$ & $0$ & $0$ \\
      $ \{ 1, 3 \} $ & $1$ & $1$ & $0$ & $0$ \\
      $ \{ 2, 3 \} $ & $0$ & $0$ & $0$ & $0$ \\
      $ \{ 1,2,3 \} $ & $1$ & $ \boldsymbol{ 1  } $ & $ \boldsymbol{ 0 } $ & $ \boldsymbol{ 0 } $
    \end{tabular}
    \bigskip 
    \caption{Decomposition of the three-player glove game as
      $ v = v _1 + v _2 + v _3 $, where $ \{ 1 \} $ and its incident edges
      are removed from the hypercube graph. This causes players $2$ and $3$ to become null players, so player $1$ receives all the value, and the game becomes inessential.
      \label{tab:holdout1Glove}}
  \end{table}
\end{example}

\begin{example}
  \label{ex:holdout2Glove}
  On the other hand, suppose that player $ \{ 2 \} $ refuses to join
  the coalition first, so that $ \{ 2 \} $ and all its incident edges
  are removed from the graph. The resulting decomposition is shown in
  \autoref{tab:holdout2Glove}. Unlike in the previous example, all
  three players still contribute marginal value on some of the
  remaining feasible edges. However, removing an edge on which player
  $2$ contributes zero marginal value causes the payoff to player $2$
  to increase from $ \frac{ 1 }{ 6 } $ to $ \frac{ 3 }{ 10 }
  $. Interestingly, player $3$ also receives a slightly increased
  payoff, from $ \frac{ 1 }{ 6 } $ to $ \frac{ 1 }{ 5 } $, since
  player $3$ contributes zero marginal value to any coalition that
  already contains player $2$, and one such coalition has been removed
  from consideration. Both players $2$ and $3$ benefit at the expense
  of player $1$, whose payoff is decreased from $ \frac{ 2 }{ 3 } $ to
  $ \frac{1}{2} $.
  \begin{table}
     \centering
     \renewcommand\arraystretch{1.4}
     \begin{tabular}{crrrr}
       $S$ & $v $ & $ v _1 $ & $ v _2 $ & $ v _3 $ \\
       \hline
       $ \emptyset $ & $0$ & $0$ & $0$ & $0 $ \\
       $ \{ 1 \} $ & $0$ & $ \frac{ 3 }{ 10 } $ & $ - \frac{ 1 }{ 10 } $ & $ - \frac{ 1 }{ 5 } $ \\
       $ \{ 3 \} $ & $0$ & $ - \frac{ 3 }{ 10 } $ & $ \frac{ 1 }{ 10 } $ & $ \frac{ 1 }{ 5 } $ \\
       $ \{ 1, 2 \} $ & $1$ & $ \frac{ 2 }{ 5 } $ & $ \frac{ 3 }{ 5 } $ & $0$ \\
       $ \{ 1, 3 \} $ & $1$ & $ \frac{1}{2} $ & $ \frac{ 1 }{ 10 } $ & $ \frac{ 2 }{ 5 } $ \\
       $ \{ 2, 3 \} $ & $0$ & $ - \frac{ 2 }{ 5 } $ & $ \frac{ 1 }{ 5 } $ & $ \frac{ 1 }{ 5 } $ \\
       $ \{ 1,2,3 \} $ & $ 1 $ & $ \boldsymbol{ \frac{1}{2} } $ & $ \boldsymbol{ \frac{ 3 }{ 10 } } $ & $ \boldsymbol{ \frac{ 1 }{ 5 } } $ 
     \end{tabular}
     \bigskip 
     \caption{Decomposition of the three-player glove game as
      $ v = v _1 + v _2 + v _3 $, where $ \{ 2 \} $ and its incident edges
      are removed from the hypercube graph.\label{tab:holdout2Glove} }
   \end{table}
 \end{example}

 \begin{remark}
   \label{rmk:precedence}
   The values obtained in \autoref{ex:holdout2Glove} are different
   from the Shapley values with precedence constraints in
   \citet{FaKe1992,KhSeTa2016}, which are
   $ ( \frac{1}{2}, \frac{ 1 }{ 4 } , \frac{ 1 }{ 4 } ) $. These take
   the approach of averaging over feasible permutations, a
   generalization of \eqref{eqn:shapleyPermutation}. In this case,
   there are four feasible permutations---$ ( 1, 2, 3 ) $,
   $ ( 1, 3, 2 ) $, $ ( 3, 1, 2 ) $, and $ ( 3, 2, 1 ) $---and player
   $1$ contributes marginal value $1$ in two of these, while players
   $2$ and $3$ each contribute marginal value $1$ in one
   permutation. This illustrates that, unlike the case in
   \autoref{sec:decomp}, these solution concepts are generally
   distinct.

   (We also note that, as pointed out in \autoref{sec:introGames}, it
   is computationally undesirable to average over permutations, since
   this is factorial in $ \lvert N \rvert $, whereas solving a
   $ \lvert V \rvert \times \lvert V \rvert $ linear system is only
   exponential in $ \lvert N \rvert $.)
 \end{remark}

 Finally, we note that we may also consider \emph{weighted} subgraphs
 of the hypercube graph, combining the approach above with that of
 \autoref{sec:weightedDecomp}, in the obvious way. The next example
 illustrates that if we weight the edges according to the degrees of
 the incident vertices---which are generally non-constant once edges
 have been removed---we may obtain different decompositions of the
 game than if we used a constant edge weight.

 \begin{example}
   Consider again the restricted glove game of
   \autoref{ex:holdout2Glove}, where player $2$ refuses to join the
   coalition first, so that $ \{ 2 \} $ and its incident edges are
   removed from the hypercube graph. Instead of taking the edge
   weights $ w \equiv 1 $, suppose we take
   $ w ( a, b ) = \mathrm{deg}(a) \mathrm{deg}(b) $. This weight
   function is related to a graph Laplacian with vertex weights
   considered by \citet{Lovasz1979}, as observed by \citet{ChLa1996}.

   The resulting decomposition is shown in
   \autoref{tab:holdout2GloveWeight}.  In this case, the players
   receive the payoffs
   $ ( \frac{1}{2} , \frac{ 1 }{ 4 } , \frac{ 1 }{ 4 } ) $, which is
   distinct from the constant-weight payoffs
   $ ( \frac{1}{2} , \frac{ 3 }{ 10 } , \frac{ 2 }{ 5 } ) $ of
   \autoref{ex:holdout2Glove}. (This also happens to agree with the
   Shapley values with precedence constraints in
   \autoref{rmk:precedence}, although this need not be the case for an
   arbitrary game.)  Note that, although $ v _2 (N) = v _3 (N) $, the
   component games $ v _2 $ and $ v _3 $ display asymmetry elsewhere:
   e.g.,
   $ v _2 \bigl( \{ 2 , 3 \} \bigr) \neq v _3 \bigl( \{ 2, 3 \} \bigr)
   $. This is due to the asymmetry between players $2$ and $3$ in the
   graph.
   
  \begin{table}
     \centering
     \renewcommand\arraystretch{1.4}
     \begin{tabular}{crrrr}
       $S$ & $v $ & $ v _1 $ & $ v _2 $ & $ v _3 $ \\
       \hline
       $ \emptyset $ & $0$ & $0$ & $0$ & $0 $ \\
       $ \{ 1 \} $ & $0$ & $ \frac{ 1 }{ 3 } $ & $ - \frac{ 1 }{ 12 } $ & $ - \frac{ 1 }{ 4 } $ \\
       $ \{ 3 \} $ & $0$ & $ - \frac{ 1 }{ 3 } $ & $ \frac{ 1 }{ 12 } $ & $ \frac{ 1 }{ 4 } $ \\
       $ \{ 1, 2 \} $ & $1$ & $ \frac{ 5 }{ 12 } $ & $ \frac{ 7 }{ 12 } $ & $0$ \\
       $ \{ 1, 3 \} $ & $1$ & $ \frac{1}{2} $ & $ \frac{ 1 }{ 12 } $ & $ \frac{ 5 }{ 12 } $ \\
       $ \{ 2, 3 \} $ & $0$ & $ - \frac{ 5 }{ 12 } $ & $ \frac{ 1 }{ 6 } $ & $ \frac{ 1 }{ 4 } $ \\
       $ \{ 1,2,3 \} $ & $ 1 $ & $ \boldsymbol{ \frac{1}{2} } $ & $ \boldsymbol{ \frac{ 1 }{ 4 }  } $ & $ \boldsymbol{ \frac{ 1 }{ 4 }  } $ 
     \end{tabular}
     \bigskip 
     \caption{Decomposition of the three-player glove game as
       $ v = v _1 + v _2 + v _3 $, where $ \{ 2 \} $ and its incident
       edges are removed from the hypercube graph, and where each edge
       is weighted by the product of the degrees of its incident
       vertices.\label{tab:holdout2GloveWeight} }
   \end{table}
 \end{example}

 \section{Conclusion}

 We have used the combinatorial Hodge decomposition on a hypercube
 graph to show that any cooperative game may be decomposed into a sum
 of component games, one component for each player, so that this
 decomposition satisfies appropriate efficiency, null-player,
 symmetry, and linearity properties. This yields a new
 characterization of the Shapley value as the value of the grand
 coalition in each player's component game.

 We have also shown that this game decomposition may be understood in
 terms of the least-squares solution to a linear problem, where the
 solution is exact if and only if the game is inessential. In this
 sense, our decomposition may be considered as an edge-based (rather
 than vertex-based) variant of the least-squares and minimum-norm
 solution concepts of \citet{RuVaZa1998} and \citet{KuSa2007}. The
 normal equations for this linear problem yield another, equivalent
 characterization of the game decomposition in terms of the
 well-studied graph Laplacian. This allowed us to obtain an explicit
 formula for each component game by applying the discrete Green's
 function for the hypercube graph Laplacian.

 Finally, we have shown how this decomposition may be generalized, in
 a natural way, using the combinatorial Hodge decomposition for
 weighted graphs and subgraphs of the hypercube graph. These
 generalized decompositions preserve the efficiency, null-player,
 symmetry (in an appropriate sense, modulo the symmetry of the weights
 and the subgraph), and linearity properties obtained earlier. This
 yields a family of decompositions, and corresponding solution
 concepts, for problems where players exhibit variable willingness or
 unwillingness to join certain coalitions, and we have compared and
 contrasted these solution concepts with those of
 \citet{FaKe1992,KhSeTa2016} for certain models of restricted
 cooperation.

 \subsection*{Acknowledgments}

 Ari Stern was supported in part by a grant from the Simons Foundation
 (\#279968). Alexander Tettenhorst was supported in part by the ARTU
 research fellowship program in the Department of Mathematics at
 Washington University in St.~Louis.

\end{document}